\documentclass[10pt,reqno,a4paper]{amsart}
\usepackage{amsmath,amssymb,dsfont,graphicx,cite,rotating,setspace}
\def\beq{\begin{equation}}
\def\eeq{\end{equation}}
\def\bea{\begin{eqnarray}}
\def\eea{\end{eqnarray}}
\newtheorem{theorem}{Theorem}
\newtheorem{rem}{Remark}
\newtheorem{prop}{Proposition}

\makeatletter
\expandafter\let\expandafter
\reset@font\csname reset@font\endcsname
\def\subeqnarray{\arraycolsep1pt
    \def\@eqnnum\stepcounter##1{\stepcounter{subequation}
        {\reset@font\rm(\theequation\alph{subequation})}}
\jot5mm     \eqnarray}

\makeatother

\def\tilde{\widetilde}
\def\hat{\widehat}

\def\su2{{\mathfrak {su}}(2)}
\def\e3{{\mathfrak {e}}(3)}

\begin{document}
\title{Linearization through symmetries for discrete equations}

\author{D. Levi  and C. Scimiterna\\
Dipartimento di Matematica e Fisica dell'Universit\'a Roma Tre  and  \\ Sezione INFN di Roma
Tre, \\
via della Vasca Navale 84, Roma, Italy 00146}

\date{\today}
\begin{abstract}
We show that one can define through the symmetry approach a procedure   to check the linearizability of a difference  equation via a point  or a discrete  Cole--Hopf transformation. If the equation is linearizable  the symmetry provides   the linearizing transformation. At the end we present few examples of applications for equations defined on four lattice points. 

\end{abstract}

\maketitle

\section{Introduction}
In recent years the use of partial difference equations has been playing an increasing role in physics and in mathematics. From one side discrete systems are believed to be at the base of the basic laws of physics (see for example the recent literature on quantum gravity \cite{dg}) and on the other side with the increasing use of computers, discretizations are playing  an increasing role  for solving numerically differential equations \cite{dc}. 

 Calogero  \cite{cz} introduced a  heuristic distinction between nonlinear partial differential equations which are  "C-integrable" and "S-integrable", namely, equations that are linearizable by an appropriate change of variables (i.e. by an explicit redefinition of the dependent variable and maybe, in some cases also  the independent variables), and those equations that are integrable via the Inverse Scattering Transform (IST) \cite{cd}.  In the cited reference Calogero used the asymptotic behavior to  find  equations belonging to these two classes of equations as it usually preserves integrability and linearizability. However this approach, even if very fruitful is often very cumbersome and not always exhaustive (see for example the results of  \cite{hlps} for discrete equations). 

A more intrinsic approach  is based on the existence of symmetries. The well-known notion of higher symmetry is at the base of this approach together with the notion of  IST.
 The mutual influence of these theories has led  to the fundamental abstract concept of formal symmetry \cite{ssm,y06},  more basic than symmetry as it provides  also higher conservation law, B\"acklund transformation and Lax pair representation. In this sense a formal
symmetry is a universal object. However the derivation of formal symmetries is difficult to apply in the case of difference equations \cite{ly}. Moreover the  distinction between $C$ and $S$ integrable equations is only at the level of the conservations laws as linearizable equations have no local conservations laws of arbitrary high order \cite{y06}.

A well established result in the framework of Lie theory for proving the linearizability of nonlinear Partial Differential Equations (PDEs) is provided by Kumei and Bluman \cite{kb} (for a recent extended review see \cite{bca}) based on the analysis of the symmetry properties of linear PDEs. Following the analogy of the continuous case  we will formulate a similar theorem for linearizable Partial Difference Equations (P$\Delta$Es) by which we recover  part of the results obtained by assuming the existence of linearizing transformations \cite{ls,ls12,ls122}. Partial results in this direction for the case of difference equations defined on a fixed lattice have been obtained by Quispel and collaborators \cite{bsq95,sbq96,sq97}.

In Section \ref{two} we  review for the sake of completeness and clarity the results presented by Bluman et. al. in \cite{bca} for the linearization of nonlinear PDEs both by point and nonlocal transformations and consider examples of their applications. Section \ref{three} is devoted to the introduction of P$\Delta$Es, their symmetry analysis and the discretization of the theorems presented in \cite{bca}.  A series of examples is presented in Section \ref{four} while in Section \ref{five} we summarize the results and present conclusions and outlooks. In  an Appendix a theorem complements the results of Section \ref{three} by showing that equations possessing an infinite dimensional symmetry algebra depending on an arbitrary function are linear and the arbitrary function satisfies a linear homogeneous equation.

\section{Linearization of PDEs through symmetries }\label{two}

In \cite{kb} Bluman and Kumei introduce a series of theorems dealing with the conditions for a nonlinear PDE to be transformable  into a linear one by {\it contact transformations}. Here, in the following, we will limit ourselves to the case of just {\it point transformations} as these are the relevant ones in the discrete case \cite{ltw, lstw}. In more recent works the same authors \cite{bk} extended the consideration to the case when we have {\it non--invertible transformations} between a nonlinear and a linear PDE. 

The basic observation is that a linear PDE, 
\bea \label{t0}
\mathcal L v(y)=\mathcal F(y),
\eea
 where $\mathcal L$ is a $v$--independent but possibly $y$--dependent linear operator and $\mathcal F(y)$ is the inhomogeneous term, has one point symmetry of  infinitesimal symmetry generator 
 \bea \label{t0b}
 \hat X = w(y)\frac{\partial}{\partial v}
 \eea
  depending on a function $w$ which satisfy the homogeneous equation 
 \bea \label{t1}
\mathcal L  w(y)=0,
\eea
 as any solution of (\ref{t0}) is the sum of a particular solution plus the general solution of the associated  homogeneous equation. 
 As the existence of an infinitesimal generator of the form (\ref{t0b}) is preserved when we  transform a linear equation into a nonlinear one by an invertible point transformation, following \cite{bca} Section 2.4, we can state in the following the  conditions for the existence of an invertible linearization mapping of a nonlinear PDE:

\begin{theorem} \label{tt1}
A nonlinear PDE 
\bea \label{t3a}
\mathcal E_n(x,u,u_x, \cdots u_{nx})=0
\eea
 of order $n$ for a scalar function $u$ of an $r$--dimensional ($r \ge 2$) vector $x$ will be linearizable by a point transformation 
\bea \label{t2} w(y)=f(x,u),\qquad  y=g(x,u) , 
\eea
 to a linear equation (\ref{t1}) for $w$ if it possesses a symmetry generator 
 \bea \label{t3}
 \hat X &=& \sum_{i=1}^r \xi_i(x,u) \partial_{x_i} + \phi(x,u) \partial_u, \quad  \xi_i(x,u)=\alpha_i(x,u) w(y), \\ \nonumber &&\quad \phi(x,u)=\omega(x,u) w(y),
 \eea
  with $\omega$ and  $\alpha_i$ given functions of their arguments and $w(y)$  an arbitrary solution of (\ref{t1}). 
\end{theorem}  
Following \cite{kb}  we can state the sufficient conditions for the existence of an invertible linearization mapping of a nonlinear PDE. This theorem  defines the transformation (\ref{t2}):

\begin{theorem} \label{tt2}
 If a symmetry generator for the nonlinear PDE (\ref{t3a}) as specified in Theorem \ref{tt1}  exists,  the  invertible transformation (\ref{t2}) which transforms (\ref{t3a}) to the linear PDE (\ref{t1}) is given by
 \bea \label{t6}
 y_i&=&\Phi_i(x,u), \quad i=1,\cdots,r, \\ \label{t7}
 w&=&\Psi(x,u).
 \eea
where $\Phi_i(x,u)$ are  $r$ functionally independent solutions, $i=1,\cdots,r$, of  the linear homogeneous first order  PDE   for a scalar function $\Phi(x,u)$ 
\bea \label{t4}
\sum_{i=1}^r \alpha_i(x,u) \Phi(x,u)_{x_i} + \omega(x,u) \Phi(x,u)_{u}=0, \\ \nonumber
\eea
 and $\Psi(x,u)$ by a particular solution of the
linear inhomogeneous first order  PDE  for a scalar function $\Psi(x,u)$ 
\bea \label{t5}
\sum_{i=1}^r \alpha_i(x,u) \Psi(x,u)_{x_i} + \omega(x,u) \Psi(x,u)_{u}=1.\\ \nonumber
\eea
\end{theorem}
If a given linearizable nonlinear PDE  does not have local 
symmetries  of the form (\ref{t3}), i.e. its local symmetries do not satisfy the criteria of Theorem
\ref{tt1},  it could still happen, as shown in  \cite{bca} Section 4.3, that a nonlocally related system has an
infinite set of local symmetries  that yields an invertible mapping of the nonlocally related system
to some linear system of PDEs. Consequently, the invertible mapping
of the nonlocally related system to a linear system will provide a nonlocal
(non-invertible) mapping of the given nonlinear PDE to a linear PDE. 
This non invertible transformation will be a kind of  Cole--Hopf transformation \cite{cole,hopf}. 
In this case, however, we have to generalize Theorem \ref{tt2} to take into account the fact that we are dealing with a system of equations.

\begin{theorem} \label{tt3}
Let us consider a  system of nonlinear PDE 
\bea \label{t3b}
\mathcal E^{(1)}_n(x,u,v,u_x, v_x,\cdots u_{nx},v_{nx})=0, \quad \mathcal E^{(2)}_n(x,u,v,u_x, v_x,\cdots u_{nx},v_{nx})=0
\eea
 of order $n$ for two scalar functions $u$ and $v$ of an $r$--dimensional ($r \ge 2$) vector $x$ which possesses a symmetry generator 
 \bea \label{t3c}
  \hat X &=& \sum_{i=1}^r \xi^i(x,u,v) \partial_{x_i} + \phi(x,u,v) \partial_u+ \psi(x,u,v) \partial_v, \quad  \xi^i(x,u,v)= \sum_{j=1}^2 \alpha^i_j(x,u,v) w^{(j)}(y), \\ \nonumber && \phi(x,u,v)= \sum_{j=1}^2\beta_j(x,u,v) w^{(j)}(y) ,  \quad \psi(x,u,v)= \sum_{j=1}^2 \gamma_j(x,u,v)w^{(j)}(y), 
 \eea
  with $\alpha_j^i$, $\beta_j$ and $\gamma_j$   given functions of their arguments and the function $w=(w^{(1)}(y), w^{(2)}(y))$ satisfying the linear homogeneous equations 
  \bea \label{t0a}
\mathcal M (y)w(y)=0,
\eea
 with $y$  an $r$--dimensional vector depending on $u$, $v$ and the vector $x$ and $\mathcal M$ is a $2x2$ matrix linear operator. 

The  invertible transformation 
\bea \label{t2a}
w^{(1)}(y)=F^{(1)}(x,u,v),\qquad  w^{(2)}(y)=F^{(2)}(x,u,v),\qquad y=G(x,u,v),
\eea which transforms (\ref{t3b}) to the system of linear PDEs  (\ref{t0a})
 is given by $r$ functionally independent solutions $G_i(x,u,v)$ with $i=1,\cdots,r$ of  the linear homogeneous first order system of PDEs  for a scalar function $G(x,u,v)$ 
\bea \label{t4a}
\sum_{i=1}^r \alpha^i_k(x,u,v) G_{x_i} + \beta_k(x,u,v) G_{u}+\gamma_k(x,u,v) G_{v}=0
\eea
 and by a particular solution of the
linear inhomogeneous first order system of PDEs for the function $F=(F^{(1)}(x,u,v),F^{(2)}(x,u,v))$ 
\bea \label{t5a}
\sum_{i=1}^r \alpha^i_k(x,u,v) F^{(j)}_{x_i} + \beta_k(x,u,v) F^{(j)}_{u}+ \gamma_k(x,u,v) F^{(j)}_{v}=\delta_k^j,
\eea
with $\delta_k^j$ the standard Kronecker symbol.
\end{theorem}

For the sake of completeness and to clarify the application of the theorems presented above, in view of the discretization which will be presented in the following section, we consider here one example of linearizable nonlinear PDEs belonging to each of the two cases presented above.

\subsection{A nonlinear PDE linearizable by a point transformation}
It is well know, see for example Olver book \cite{olver}, that the potential Burgers equation 
\bea \label{b1}
u_t = u_{xx} + (u_x)^2,
\eea
 is linearizable by a point transformation. In fact the infinite dimensional part of the infinitesimal generator of its point symmetries is given by
 \bea \label{b2}
 \hat X =  v(x,t) e^{-u} \partial_u,
 \eea
 where  $v(x,t)$ satisfies the homogeneous linear heat  equation $v_t-v_{xx}=0$.
 
 The conditions of Theorem \ref{tt1} are satisfied with  $\omega=e^{-u}$ and $\alpha_i=0$.
 We can apply Theorem \ref{tt2} and we  get $\Phi_1=x$ and $\Phi_2=t$  as from (\ref{t4}) $\Phi_u=0$ while from   (\ref{t5}) $\Psi(x,u)$ satisfies the equation $  \Psi(x,u)_{u}=e^{u}$ i.e. 
 \bea \label{b3}
 u=\log_e(w).
 \eea
 Eq. (\ref{b3}) is the linearizing transformation for the potential Burgers equation (\ref{b1}).
\subsection{A nonlinear PDE linearizable by a non invertible transformation}
The standard example in this class is the Burgers equation
\bea \label{b4}
u_t = u_{xx} - u u_x=[u_x - \frac{1}{2} u^2]_x,
\eea 
linearizable by a Cole--Hopf transformation. As (\ref{b4}) has no infinite dimensional symmetry algebra but it is written as a conservation law we can introduce a potential function $v(x,t)$ and (\ref{b4}) can be written as the system 
\bea \label{b5}
v_x&=&2u, \\ \nonumber
v_t&=& 2 u_x - u^2.
\eea
Applying Theorem \ref{tt3}   we can find an infinite dimensional  symmetry of the form of (\ref{t3c}). In fact, solving the determining equations, apart from  terms corresponding to a finite dimensional algebra, we obtain an infinite dimensional dilation symmetry given by
\bea \label{b7}
\psi=4 \omega(x,t) e^{\frac{v}{4}}, \qquad \phi=\frac{1}{2} \psi_x + u \psi_v=e^{\frac{v}{4}} [ 2 \omega_x + \omega u]
\eea
where $\omega(x,t)$ satisfies the linear heat equation $\omega_{t} - \omega_{xx}=0$. 

The linearizing transformation can be obtained from Theorem \ref{tt3}. Let us define $w^{(1)}(y)=\omega(x,t)$, $w^{(2)}(y)=\omega_x(x,t)$ and take as functionally independent solutions of (\ref{t4a}) $G_1=x$ and $G_2=t$. 
 As $\alpha_k^i=0$ and $\gamma_1=4 e^{\frac{v}{4}}$, $\gamma_2=0$, $\beta_1=u e^{\frac{v}{4}}$ and $\beta_2=2 e^{\frac{v}{4}} $, we get as a particular solution of (\ref{t5a})
\bea \label{b8}
F^{(1)}=-e^{-\frac{v}{4}}, \quad F^{(2)}=\frac{1}{2} u e^{-\frac{v}{4}}.
\eea
Eq. (\ref{t2a}) implies $ \omega=-e^{-\frac{v}{4}}$ and $\omega_x=\frac{1}{2} u e^{-\frac{v}{4}}$ and from it we obtain as a linearizing transformation
the Cole--Hopf transformation 
\bea \nonumber
u=-2 \frac{\omega_x}{\omega}.
\eea

\section{P$\Delta$Es and their linearization}\label{three}
We will consider here the problem of the linearization of a partial difference equation. 

As is well known difference equations can be obtained  on a given grid, starting from physical or chemical or biological problems dealing with lattice systems,  or as discretization of differential equations with symmetries which we want to preserve. In this second case we can use the freedom of the grid to  get a difference scheme which preserve  part of the symmetries of the continuous equation.

In this first approach to the problem of linearization of difference equation we will consider  the case when the grid is preassigned and assumed to be fixed, with constant  lattice spacing. Moreover for simplicity we will consider autonomous equations  defined on a two dimensional grid  so that there is no privileged position and we can write the dependent variables just in terms of the shifts with respect to the reference point  $u_{n,m}=u_{0,0}$ on the lattice.

A partial difference equation (P$\Delta$E) of order $N \cdot N'$ for a function $u_{n,m}$ will be a relation between $N \cdot N'$ points in the two dimensional grid, i.e.
\bea \label{e1}
\mathcal E_{N \cdot N'} \Big (u_{0,0}, u_{1,0},  \cdots, u_{N,0}, u_{0,1}, \cdots,u_{N,1}, \cdots, u_{N,N'} \Big )=0.
\eea

A continuous symmetry for  equations of the form (\ref{e1}), where the lattice is fixed, i.e. the two independent variables $x_{n,m}$ and $t_{n,m}$ are completely specified as $x_{n,m}=h_x n+x_0$ and $t_{n,m}=h_t m + t_0$  with $h_x$, $h_t$, $x_0$ and $t_0$ given constants, is given just by dilations 
\bea \label{e2}
\hat X_{n,m} = \chi_{n,m}(u_{n,m})\partial_{u_{n,m}}.
\eea
It is easy to show  that  a linear P$\Delta$E of order$N\cdot N'$ for a function $v_{n,m}$
\bea \label{e1a}
\mathcal F_{N \cdot N'}=b(n,m)+\sum_{(i,j)=(0,0)}^{(N,N')} a^{i,j}(n,m) v_{n+i,m+j}=0,
\eea
has always the symmetry
 \bea \label{e4}
\hat X_{n,m} = \phi_{n,m} \partial_{v_{n,m}},
\eea
where  $\phi_{n,m} $ is a solution of the homogeneous part of (\ref{e1a}). It is not at all obvious, however,  that an equation (\ref{e1}) having a symmetry (\ref{e4}) is  linear and that the function $\phi_{n,m}$ must satisfy a homogeneous  linear equation. We leave to the Appendix the proof of this proposition.
The symmetry  (\ref{e4}) is due to the superposition principle for linear equations. If the nonlinear equation (\ref{e1}) is linearizable by a point transformation then the symmetry (\ref{e4}) must be preserved. This is the content of the Theorem \ref{tt1} we presented in the previous Section in the case of PDEs and this must still be valid here. So we can state the following theorem:

\begin{prop} \label{dt1}
An autonomous linear P$\Delta$E (\ref{e1}) is linearizable if it has a point symmetry of the form 
\bea \label{e4a}
\hat X_{n,m} = \alpha_{n,m}(u_{n,m})  \phi_{n,m} \partial_{v_{n,m}},
\eea
 where the function $\phi_{n,m}$ satisfies a linear P$\Delta$E of the form (\ref{e1a}) with $b(n,m)=0$.
\end{prop} 

The proof of Theorem \ref{dt1} and of the following Theorem \ref{dt2} are  the same as those presented by Bluman and Kumei \cite{bk} in the continuous case  and so we will not repeat them here.

As in the case of PDEs we can present the following  theorem which provide the transformation which reduces the equation to a linear one. In this case, as the independent variables are not changed,  the transformation is given  just  by a dilation. So we have:

\begin{theorem} \label{dt2}
The point transformation which linearizes the nonlinear P$\Delta$E (\ref{e1})
\bea \label{e5}
v_{n,m} = \Psi_{n,m} (u_{n,m})
\eea
 is obtained by solving the differential equation 
 \bea \label{e6}
 \alpha_{n,m}(u_{n,m}) \frac{d  \Psi_{n,m} (u_{n,m})}{d u_{n,m}} = 1.
\eea
\end{theorem}

As in the continuous case, if (\ref{e1}) has no symmetries of the form considered in Theorem \ref{dt1} we can introduce some potential variables. On the lattice there are infinitely many ways to introduce a potential variable as there are infinitely many  ways to define a first derivative. Thus it seems to be advisable to check the equation with a linearizability criterion like the algebraic entropy \cite{v} before looking for potential variables.

The simplest way to introduce a potential symmetry is by writing the difference equation (\ref{e1}) as a system 
\bea \label{e7}
v_{n+1,m}= \mathcal E^{(1)}_{n,m}(u_{n,m}, \cdots ), \qquad v_{n,m+1}=\mathcal E^{(2)}_{n,m}(u_{n,m}, \cdots).
\eea
In such a way  
\bea \label{e8}
\mathcal E_{N \cdot N'} \Big (u_{0,0}, u_{1,0},  \cdots, u_{N,0}, u_{0,1}, \cdots,u_{N,1}, \cdots, u_{N,N'})= \mathcal E^{(1)}_{n,m+1}-\mathcal E^{(2)}_{n+1,m}.
\eea
However it is easy to show in full generality that the symmetries for (\ref{e7}) and for (\ref{e1}) are the same.

We  can introduce potential symmetries by the following system,  
\bea \label{e7a}
v_{n+1,m} - v_{n,m}= \mathcal E^{(1)}_{n,m}(u_{n,m}, \cdots ), \qquad v_{n,m+1}- v_{n,m}=\mathcal E^{(2)}_{n,m}(u_{n,m}, \cdots).
\eea
In such a way 
\bea \label{e8a}
\mathcal E_{N \cdot N'} \Big (u_{0,0}, u_{1,0},  \cdots, u_{N,0}, u_{0,1}, \cdots,u_{N,1}, \cdots, u_{N,N'})=[ \mathcal E^{(1)}_{n,m+1}-\mathcal E^{(1)}_{n,m}] - [\mathcal E^{(2)}_{n+1,m}-\mathcal E^{(2)}_{n,m}].
\eea
i.e. the nonlinear difference equation  is written as a discrete conservation law.
As (\ref{e7a}) is a system, to construct the symmetries  we have to generalize  the linearization theorem as we did in the continuous case.

\begin{theorem} \label{dt3}
Let us consider a  system of nonlinear P$\Delta$Es 
\bea \label{e9}
\mathcal F^{(1)}_{n,m}(u_{n,m}, \cdots, v_{n,m}, \cdots)=0, \quad \mathcal F^{(2)}_{n,m}(u_{n,m}, \cdots, v_{n,m}, \cdots)=0
\eea
 of order $N\cdot N'$ for two scalar functions $u_{n,m}$ and $v_{n,m}$ of two indices $n$ and $m$ which possesses a symmetry generator 
 \bea \label{e10}
  \hat X &=&  \phi_{n,m}(u_{n,m},v_{n,m}) \partial_{u_{n,m}}+ \psi_{n,m}(u_{n,m},v_{n,m}) \partial_{v_{n,m}}, \\ \nonumber && \phi_{n,m}(u_{n,m},v_{n,m})= \sum_{j=1}^2 \beta^{(j)}_{n,m}(u_{n,m},v_{n,m}) w^{(j)}_{n,m} ,  \quad \psi_{n,m}(u_{n,m},v_{n,m})=\sum_{j=1}^2 \gamma^{(j)}_{n,m}(u_{n,m},v_{n,m})w^{(j)}_{n,m}, 
 \eea
  with  $\beta^{(j)}$ and $\gamma^{(j)}$   given functions of their arguments and the function $w=(w^{(1)}_{n,m}, w^{(2)}_{n,m})$ satisfying the linear equations 
  \bea \label{e11}
\mathcal L _{n,m}w_{n,m}=0.
\eea

The  invertible transformation 
\bea \label{e12}
w^{(1)}_{n,m}=K^{(1)}_{n,m}(u_{n,m},v_{n,m}),\qquad  w^{(2)}_{n,m}=K^{(2)}_{n,m}(u_{n,m},v_{n,m}),
\eea which transforms (\ref{e9}) to the system of linear PDEs  (\ref{e11})
 is given by  a particular solution of the
linear inhomogeneous first order system of PDEs for the function $K=(K^{(1)}_{n,m}(u_{n,m},v_{n,m}),K^{(2)}_{n,m}(u_{n,m},v_{n,m}))$ 
\bea \label{e13}
 \beta^{(k)}_{n,m}(u_{n,m},v_{n,m}) K^{(j)}_{n,m}+ \gamma^{(k)}_{n,m}(u_{n,m},v_{n,m}) K^{(j)}_{n,m}=\delta_k^j,
\eea
where $\delta_k^j$ is the standard Kronecker symbol.
\end{theorem}

\section{Examples.} \label{four}
Here we present a few examples of linearizable P$\Delta$Es. For concreteness and for comparing with the previous literature \cite{ls} we limit ourselves to the case when the nonlinear difference equation involves at most 4 lattice points. 

\subsection{Classification of  P$\Delta$Es on a square lattice linearizable by point transformations}

We consider a general autonomous equation defined on a square lattice:
\bea \label{f1}
\mathcal F(u_{0,0}, u_{0,1}, u_{1,0}, u_{1,1})=0,
\eea
where, for convenience,  we just write down the shift with respect to the reference point of indices $n,m$. If we assume that $u_{1,1}$ is present in (\ref{f1}) then we can rewrite the equation  as
\bea \label{f2}
u_{1,1}=F(u_{0,0}, u_{0,1}, u_{1,0}).
\eea
Following Theorem \ref{dt1} we look for an infinitesimal symmetry generator of the form
\bea \label{f3}
\hat X_{0,0} =  \alpha_{0,0}(u_{0,0}) \phi_{0,0} \partial_{u_{0,0}}, 
\eea
where the function $\phi$ solves a linear homogeneous equation, i.e.
\bea \label{f4}
\hat {\mathcal L} \phi_{0,0} = 0, \qquad  \hat {\mathcal L} = a + b T_1 + c T_2 + d T_1 T_2,
\eea
with $T_1$ and $T_2$ operators such that $T_1 \phi_{0,0} = \phi_{1,0}$ and $T_2 \phi_{0,0} = \phi_{0,1}$. If $d \ne 0$ then we can write (\ref{f4}) as
\bea \label{f5}
\phi_{1,1}= -\frac{1}{d} [a \phi_{0,0} + b \phi_{1,0} + c \phi_{0,1}].
\eea
In this setting  $u_{0,i}$, $u_{j,0}$, $\phi_{0,i}$ and $\phi_{j,0}$, with $i,j=0,1$ are independent variables. 
If (\ref{f3}) is a  generator of the symmetries of (\ref{f2}) then we must have 
\bea \label{f6}
\hat X \mathcal F &=& 0 \, \leftrightarrow \, F_{,u_{0,0}} \phi_{0,0} \alpha_{0,0}(u_{0,0}) + F_{,u_{1,0}} \phi_{1,0} \alpha_{1,0}(u_{1,0}) + F_{,u_{0,1}} \phi_{0,1} \alpha_{0,1}(u_{0,1})= \\ \nonumber &&= \phi_{1,1} \alpha_{1,1}(u_{1,1}) =  -\frac{1}{d} [ a \phi_{0,0} + b \phi_{1,0} + c \phi_{0,1}] \alpha_{1,1}(F(u_{0,0}, u_{0,1}, u_{1,0})).
  \eea
As $\phi_{0,0}$, $\phi_{1,0}$ and $\phi_{0,1}$ are independent variables, we obtain from (\ref{f6}) three equations relating the  function $\alpha$, intrinsic of the symmetry, with the function $F$, intrinsic of the nonlinear equation:
\bea \label{f7}
a \alpha_{1,1}(F) + d F_{,u_{0,0}} \alpha_{0,0}(u_{0,0}) =0, \quad b \alpha_{1,1}(F) + d F_{,u_{1,0}} \alpha_{1,0}(u_{1,0}) =0, \quad c \alpha_{1,1}(F) + d F_{,u_{0,1}} \alpha_{0,1}(u_{0,1}) =0.
\eea
As in (\ref{f7}), up to a constant,  the first term is the same for all three equations, we can rewrite them as a  system of PDE's for the function $F$ depending on $\alpha$
\bea \label{f8}
\frac{1}{a} F_{,u_{0,0}} \alpha_{0,0}(u_{0,0} )= \frac{1}{b} F_{,u_{1,0}} \alpha_{1,0}(u_{1,0}) = \frac{1}{c} F_{,u_{0,1}} \alpha_{0,1}(u_{0,1}),
\eea
which can be solved on the characteristic, giving $F$ as a function of the symmetry variable 
\bea \label{f9}
\xi = a g(u_{0,0}) + b g(u_{1,0}) + c g(u_{0,1}), \qquad \alpha(x)=\frac{1}{g_{,x}(x)}.
\eea
Introducing this result in Theorem \ref{dt2} we get $\psi(u_{0,0})=\int^{u_{0,0}} g_x(x) dx =g(u_{0,0})+\kappa$, with $\kappa$ an arbitrary integration constant.
Then (\ref{f7}) gives that any linearizable nonlinear P$\Delta$E on a four points lattice must be written as 
\bea \label{f10}
d F_{,\xi}  + \alpha(F(\xi))=0 \, \rightarrow \, F= g^{-1}(\frac{\xi - \xi_0}{d}),
\eea
where by $g^{-1}(x)$ we mean the inverse of the function $g(x)$ given in (\ref{f9}).

Let us notice that from \eqref{f8} we can get the six linearizability necessary conditions we introduced in \cite{ls12} to classify linearizable, multilinear equations on the four lattice points, that is
\begin{subequations}\label{Osiris}
\bea
A\left(x,u_{0,1}\right)\doteq\frac{F_{,u_{0,0}}}{F_{,u_{1,0}}}\vert_{u_{0,0}=u_{1,0}=x}=\frac{a}{b},\ \ \ \forall x,\ u_{0,1},\label{Osiris1}\\
B\left(x,u_{1,0}\right)\doteq\frac{F_{,u_{0,0}}}{F_{,u_{0,1}}}\vert_{u_{0,0}=u_{0,1}=x}=\frac{a}{c},\ \ \ \forall x,\ u_{1,0},\label{Osiris2}\\
C\left(x,u_{0,0}\right)\doteq\frac{F_{,u_{0,1}}}{F_{,u_{1,0}}}\vert_{u_{1,0}=u_{0,1}=x}=\frac{c}{b},\ \ \ \forall x,\ u_{0,0},\label{Osiris3}\\
\frac{\partial}{\partial u_{0,1}}\frac{F_{,u_{0,0}}}{F_{,u_{1,0}}}=0,\ \ \ \forall u_{0,0},\ u_{1,0},\ u_{0,1},\label{Osiris4}\ \ \ \ \ \ \ \ \\
\frac{\partial}{\partial u_{1,0}}\frac{F_{,u_{0,0}}}{F_{,u_{0,1}}}=0,\ \ \ \forall u_{0,0},\ u_{1,0},\ u_{0,1},\label{Osiris5}\ \ \ \ \ \ \ \ \\
\frac{\partial}{\partial u_{0,0}}\frac{F_{,u_{0,1}}}{F_{,u_{1,0}}}=0,\ \ \ \forall u_{0,0},\ u_{1,0},\ u_{0,1}.\label{Osiris6}\ \ \ \ \ \ \ \ 
\eea
\end{subequations}
So linearizable equations on four lattice points are characterized by a function $g(x)$ and its inverse.  As a trivial example we can choose $g(x)=e^x$ and we get that the nonlinear equation $u_{1,1}=\log(\alpha e^{u_{0,0}} + \beta e^{u_{1,0}} + \gamma e^{u_{0,1}} + k)$ linearizes to $\psi_{1,1}=\alpha \psi_{0,0}+\beta \psi_{1,0} + \gamma \psi_{0,1}$. The corresponding function $\alpha$ is $\alpha(x)=e^{-x}$ and the linearizing transformation is $\psi_{0,0}=e^{u_{0,0}} + \kappa$.

In \cite{ls} we have shown that there is a multilinear equation on the square lattice belonging to the $Q^+$ class which is linearizable. It is interesting to find  the corresponding function $g(x)$ in term of which we can linearize it. The function $F$ in this case is a fraction of a second order polynomial over a third order polynomial. The only function $g$ which provides this structure is $g(x)=\frac{1}{\ell_1 x+\ell_0}$ which gives $F=-\frac{1}{\ell_1}\Big [ \frac{d}{\xi -  \xi_0} + \ell_0 \Big ]$ where $\xi=\frac{a}{\ell_1 u_{0,0} + \ell_0} +\frac{b}{\ell_1 u_{1,0} + \ell_0} + \frac{c}{\ell_1 u_{0,1} + \ell_0}$. In this case the linearizing transformation is the linear fractional function $\psi_{0,0}(u_{0,0})=\frac{1}{\ell_1 u_{0,0} +\ell_0} + \kappa$, where $\kappa$ is an arbitrary constant.

\subsection{Linearizable potential equations}
 For the sake of simplicity we set in (\ref{e7a}) $\mathcal E^{(2)}_{n,m}=u_{n,m}$. If we want the equation (\ref{e1}) to be on the square we have to choose  $\mathcal E^{(1)}_{n,m}=g_{n,m}(u_{n,m}, u_{n+1,m})$.  The application of the prolongation of the infinitesimal generator (\ref{e10}) to the second equation in (\ref{e7a}) gives \cite{lw}
 \bea \label{g1}
 \phi_{0,0}(u_{0,0},v_{0,0})=\psi_{0,1}(u_{0,1},v_{0,1})-\psi_{0,0}(u_{0,0},v_{0,0}), \quad \rightarrow  \psi_{0,0}(u_{0,0},v_{0,0})=\psi_{0,0}(v_{0,0}).
 \eea
Then the prolongation of the infinitesimal generator (\ref{e10}) applied to the first equation   in (\ref{e7a}) gives
\bea \label{g2}
\psi_{1,0}(v_{1,0})-\psi_{0,0}(v_{0,0})=[ \psi_{0,1}(v_{0,1}) - \psi_{0,0}(v_{0,0})]\frac{\partial g_{0,0}}{\partial u_{0,0}} + [ \psi_{1,1}(v_{1,1}) - \psi_{1,0}(v_{1,0})]\frac{\partial g_{0,0}}{\partial u_{1,0}},
\eea
where
\bea \label{g3}
v_{1,0}=v_{0,0}+g_{0,0}(u_{0,0}, u_{1,0}), \quad v_{0,1}=v_{0,0}+u_{0,0}, \quad v_{1,1}=v_{0,0}+u_{1,0}+g_{0,0}(u_{0,0}, u_{1,0}).
\eea
To comply with Theorem \ref{dt3} we look for an infinitesimal coefficient of the infinitesimal generator (\ref{e10}) of the form
\bea \label{g4}
\psi_{0,0}=w^{(1)}_{0,0} \gamma^{(1)}_{0,0}(v_{0,0}) + w^{(2)}_{0,0} \gamma^{(2)}_{0,0}(v_{0,0}),
\eea
where the functions $w^{(1)}_{n,m}$ and $w^{(2)}_{n,m}$ satisfy a linear partial difference equation on the square
\bea \label{g5}
w^{(1)}_{0,0} &=& a^{(1)}_{0,0} w^{(1)}_{0,1} + a^{(2)}_{0,0} w^{(1)}_{1,0} + a^{(3)}_{0,0} w^{(1)}_{1,1}, \\ \nonumber
w^{(2)}_{0,0} &=& b^{(1)}_{0,0} w^{(2)}_{0,1} + b^{(2)}_{0,0} w^{(2)}_{1,0} + b^{(3)}_{0,0} w^{(2)}_{1,1}.
\eea
Introducing (\ref{g4}, \ref{g5}) into (\ref{g2}) and taking into account that we can always choose $w^{(1)}_{0,1}$, $w^{(1)}_{1,0}$, $w^{(1)}_{1,1}$, $w^{(2)}_{0,1}$, $w^{(2)}_{1,0}$ and  $w^{(2)}_{1,1}$ as independent variables we get the following system of coupled equations for $\gamma^{(1)}$
\bea  \label{g6a}
&\gamma^{(1)}_{1,0}(v_{0,0} + g_{0,0}) \big ( 1 +  \frac{\partial g_{0,0}}{\partial u_{1,0}} \big ) - a^{(2)}_{0,0} \gamma^{(1)}_{0,0}(v_{0,0})  \big ( 1 +  \frac{\partial g_{0,0}}{\partial u_{0,0}} \big ) &= 0, \\ \label{g6b}
&\gamma^{(1)}_{0,1}(v_{0,0} + u_{0,0})  \frac{\partial g_{0,0}}{\partial u_{0,0}}  + a^{(1)}_{0,0} \gamma^{(1)}_{0,0}(v_{0,0})  \big ( 1 +  \frac{\partial g_{0,0}}{\partial u_{0,0}} \big ) &= 0, \\ \label{g6c}
&\gamma^{(1)}_{1,1}(v_{0,0} + u_{1,0} + g_{0,0})   \frac{\partial g_{0,0}}{\partial u_{1,0}}  + a^{(3)}_{0,0} \gamma^{(1)}_{0,0}(v_{0,0})  \big ( 1 +  \frac{\partial g_{0,0}}{\partial u_{0,0}} \big ) &= 0,
\eea
and  similar ones for the function $\gamma^{(2)}_{n,m}(v_{0,0})$.
Adding (\ref{g6a}) multiplied by $a^{(3)}_{0,0}$ to (\ref{g6c}) multiplied by $a^{(1)}_{0,0}$ we get
\bea \label{g6d}
a^{(3)}_{0,0} \gamma^{(1)}_{1,0}(v_{0,0} + g_{0,0}) \big ( 1 +  \frac{\partial g_{0,0}}{\partial u_{1,0}} \big ) + a^{(2)}_{0,0} \gamma^{(1)}_{1,1}(v_{0,0} + u_{1,0} + g_{0,0})   \frac{\partial g_{0,0}}{\partial u_{1,0}}  = 0,
\eea
an equation  similar to (\ref{g6b}), i.e  upshifting  by one the first index in (\ref{g6b}) and comparing the  result with (\ref{g6d}) we get a discrete equation for $g_{n,m}$
\bea \label{g7}
\frac{a^{(3)}_{0,0}}{a^{(1)}_{1,0}} \frac{\big ( 1 +  \frac{\partial g_{0,0}}{\partial u_{1,0}} \big )}{\big ( 1 +  \frac{\partial g_{1,0}}{\partial u_{1,0}} \big )} = a^{(2)}_{0,0} \frac{\frac{\partial g_{0,0}}{\partial u_{1,0}}}{\frac{\partial g_{1,0}}{\partial u_{1,0}}}.
\eea
In (\ref{g7})  appears the function $g_{1,0}=g_{1,0}(u_{1,0}, u_{2,0})$ and if $\frac{\partial^2 g_{1,0}}{\partial u_{1,0}\partial u_{2,0}} \ne 0$ we get a linear differential equation for $g_{0,0}$ whose solution is $g_{0,0}=g_{0,0}^{(2)}(u_{0,0})+ g_{0,0}^{(1)}(u_{0,0}) u_{1,0}$. Introducing this solution in (\ref{g7}) we get   $g_{0,0}=g_{0,0}^{(0)}+ g_{0,0}^{(1)}(u_{0,0}) u_{1,0}+g_{0,0}^{(2)} u_{0,0} $, i.e. a linear equation. By choosing, in place of (\ref{g5}) the most general linear coupled system of difference equations on the square lattice for $w^{(1)}$ and $w^{(2)}$, we would get the same result. So the introduced potential  equation (\ref{e7a}) does not provide linearizable discrete equations.

A discrete linearizable Burgers equation has been presented by Levi, Ragnisco and Bruschi considering  B\"acklund transformation of the Burgers hierarchy \cite{lrb} and by Heredero, Levi and Winternitz \cite{hlw1,hlw2}. In \cite{lrb} we can find the discrete equation 
\bea \label{g8}
u_{m+1,n} [ p + u_{m+1,n+1}] - u_{m,n} [ p + u_{m+1,n}]=0,
\eea
 and its Lax pair
 \bea \label{g9}
 \psi_{m,n+1}= u_{m,n} \psi_{m,n}, \qquad \psi_{m+1,n} = \frac{1}{p + u_{m+1,n}} \psi_{m,n}.
 \eea
 The linear P$\Delta$E corresponding to (\ref{g8}) is 
 \bea \label{g8a}
 \psi_{m+1,n+1} = \psi_{m,n} - p \psi_{m+1,n}.
 \eea
 Eqs. (\ref{g9}) suggest to rewrite (\ref{e7a}) as 
 \bea \label{g10}
 v_{n+1,m}/v_{n,m} = \mathcal E^{(1)}_{m,n}(u_{n,m}, \cdots ), \qquad  v_{n,m+1}/v_{n,m} = \mathcal E^{(2)}_{n,m}(u_{n,m}, \cdots ).
 \eea
 Eqs. (\ref{e7a}, \ref{g10}) are transformable one into the other by defining $v_{n,m}=\log(w_{n,m})$ and redefining appropriately the functions $\mathcal E^{(1)}_{n,m}$ and $\mathcal E^{(2)}_{n,m}$. However in doing so, if $w_{n,m}$ satisfies a linear equation, this will not be the case for $v_{n,m}$. So the fact that the ansatz (\ref{e7a}) do not give rise to  linearizable equations is not in contradiction with the fact that (\ref{g8}) is linearizable.
 
 The compatibility of (\ref{g10}) implies
 \bea \label{g11}
 \mathcal E^{(1)}_{n,m+1}(u_{n,m+1}, \cdots ) \mathcal E^{(2)}_{n,m}(u_{n,m}, \cdots )=\mathcal E^{(2)}_{n+1,m}(u_{n+1,m}, \cdots )\mathcal E^{(1)}_{n,m}(u_{n,m}, \cdots ).
 \eea
 If (\ref{g11}) is constrained to be an equation on the square lattice, then we must have $\mathcal E^{(1)}_{n,m}(u_{n,m}, \cdots )=\mathcal E^{(1)}_{n,m}(u_{n,m}, u_{n+1,m} )$ and $\mathcal E^{(2)}_{n,m}(u_{n,m}, \cdots )=\mathcal E^{(2)}_{n,m}(u_{n,m}, u_{n,m+1} )$. Moreover with no loss of generality we can set $\mathcal E^{(2)}_{n,m}(u_{n,m}, u_{n,m+1} )= u_{n,m}$.
 
 Let us look for the symmetries of (\ref{g10}). Applying the infinitesimal generator (\ref{e10}) to the right hand equation in (\ref{g10}) we get 
 \bea \label{g12}
 \psi_{0,0}=\psi_{0,0}(v_{0,0}), \quad 
 \phi_{0,0}=\frac{\psi_{0,1}(v_{0,1})-u_{0,0} \psi_{0,0}(v_{0,0})}{v_{0,0}}.
 \eea
 Then the determining equation associated to the left hand equation in (\ref{g10}) is given by
 \bea \label{g13}
 \psi_{1,0}(v_{1,0}) = \Big [ \frac{\partial \mathcal E^{(1)}_{0,0}}{\partial u_{0,0}} \phi_{0,0} +  \frac{\partial \mathcal E^{(1)}_{0,0}}{\partial u_{1,0}} \phi_{1,0} \Big ] v_{0,0} + \mathcal E^{(1)}_{0,0} \psi_{0,0}(v_{0,0}),
 \eea
 where the functions $\phi_{i,j}$ are expressed in term of the functions $\psi_{i,j}$ through (\ref{g12}).
 
 As we look for linearizable equations,  from Theorem \ref{dt3} it follows that we must have:
 \bea \label{g14}
 \psi_{0,0}(v_{0,0}) = \sum_{j=1}^2 w^{(j)}_{0,0} \gamma_{0,0}^{(j)}(v_{0,0}),
 \eea
 where the discrete functions $w^{(j)}_{0,0}$ satisfies a linear difference equation on the square. We can assume that 
 the coefficient of $w^{(j)}_{1,1}$ is always different from zero so that we have
 \bea \label{g15}
 w^{(1)}_{1,1} &=& a^{(1)} w^{(1)}_{0,0} + b^{(1)} w^{(1)}_{0,1} + c^{(1)} w^{(1)}_{1,0}+d^{(1)} w^{(2)}_{0,0} + e^{(1)} w^{(2)}_{0,1} + f^{(1)} w^{(2)}_{1,0}, \\ \nonumber
 w^{(2)}_{1,1} &=& a^{(2)} w^{(1)}_{0,0} + b^{(2)} w^{(1)}_{0,1} + c^{(2)} w^{(1)}_{1,0}+d^{(2)} w^{(2)}_{0,0} + e^{(2)} w^{(2)}_{0,1} + f^{(2)} w^{(2)}_{1,0}.
 \eea
 In such a case the variables $w^{(j)}_{0,0}$, $w^{(j)}_{1,0}$ and $w^{(j)}_{0,1}$, $j=1,2$, are independent and (\ref{g13}) splits in three couples of equations relating the functions $\gamma^{(j)}_{0,0}(v_{0,0})$,  $j=1,2$, with the function $\mathcal E^{(1)}_{0,0}$
 \bea \label{g16}
 && \gamma^{(1)}_{1,0}  \mathcal E^{(1)}_{0,0} =  \frac{\partial \mathcal E^{(1)}_{0,0}}{\partial u_{1,0}} \Big [ b^{(1)} \gamma^{(1)}_{1,1} + b^{(2)} \gamma^{(2)}_{1,1} - u_{1,0} \gamma^{(1)}_{1,0} \Big ],\\ \label{g16a}
 && \gamma^{(2)}_{1,0}  \mathcal E^{(1)}_{0,0} =  \frac{\partial \mathcal E^{(1)}_{0,0}}{\partial u_{1,0}} \Big [ e^{(1)} \gamma^{(1)}_{1,1} + e^{(2)} \gamma^{(2)}_{1,1} - u_{1,0} \gamma^{(2)}_{1,0} \Big ],\\ \label{g17}
 && \mathcal E^{(1)}_{0,0} \frac{\partial \mathcal E^{(1)}_{0,0}}{\partial u_{0,0}}  \gamma^{(1)}_{0,0} u_{0,0} =  \frac{\partial \mathcal E^{(1)}_{0,0}}{\partial u_{1,0}} \Big [ a^{(1)} \gamma^{(1)}_{1,1} + a^{(2)} \gamma^{(2)}_{1,1}  \Big ] + \Big [  \mathcal E^{(1)}_{0,0} \Big ]^2 \gamma^{(1)}_{0,0} , \\ \label{g17a}
  && \mathcal E^{(1)}_{0,0} \frac{\partial \mathcal E^{(1)}_{0,0}}{\partial u_{0,0}}  \gamma^{(2)}_{0,0} u_{0,0} =  \frac{\partial \mathcal E^{(1)}_{0,0}}{\partial u_{1,0}} \Big [ d^{(1)} \gamma^{(1)}_{1,1} + d^{(2)} \gamma^{(2)}_{1,1}  \Big ] + \Big [  \mathcal E^{(1)}_{0,0} \Big ]^2 \gamma^{(2)}_{0,0} , \\ \label{g18}
   && \mathcal E^{(1)}_{0,0} \frac{\partial \mathcal E^{(1)}_{0,0}}{\partial u_{0,0}}  \gamma^{(1)}_{0,1}  +  \frac{\partial \mathcal E^{(1)}_{0,0}}{\partial u_{1,0}} \Big [ c^{(1)} \gamma^{(1)}_{1,1} + c^{(2)} \gamma^{(2)}_{1,1}  \Big ]  =0, \\ \label{g18a}
  && \mathcal E^{(1)}_{0,0} \frac{\partial \mathcal E^{(1)}_{0,0}}{\partial u_{0,0}}  \gamma^{(2)}_{0,1}  +  \frac{\partial \mathcal E^{(1)}_{0,0}}{\partial u_{1,0}} \Big [ f^{(1)} \gamma^{(1)}_{1,1} + f^{(2)} \gamma^{(2)}_{1,1}  \Big ]  =0,
 \eea
 where $v_{0,1}=u_{0,0} v_{0,0}$, $v_{1,0}=\mathcal E^{(1)}_{0,0} v_{0,0}$ and $v_{1,1}=u_{1,0} \mathcal E^{(1)}_{0,0} v_{0,0}$ due to  (\ref{g10}).

 As $\mathcal E^{(1)}_{0,0}$ is a function of $u_{0,0}$ and $\gamma^{(j)}_{0,0}$ is a function of $v_{0,0}$ we get from (\ref{g16}, \ref{g16a})
 \bea \label{g19}
 \frac{ \mathcal E^{(1)}_{0,0}+u_{1,0} \frac{\partial \mathcal E^{(1)}_{0,0}}{\partial u_{1,0}}}{ \frac{\partial \mathcal E^{(1)}_{0,0}}{\partial u_{1,0}}}= \kappa_0 = \frac{b^{(1)} \gamma^{(1)}_{1,1} + b^{(2)} \gamma^{(2)}_{1,1} - u_{1,0} \gamma^{(1)}_{1,0}}{\gamma^{(1)}_{1,0} }= \frac{e^{(1)} \gamma^{(1)}_{1,1} + e^{(2)} \gamma^{(2)}_{1,1} - u_{1,0} \gamma^{(2)}_{1,0} }{ \gamma^{(2)}_{1,0} }.
 \eea
 From (\ref{g17}, \ref{g17a}) we get 
 \bea \label{g20}
 \frac{ \mathcal E^{(1)}_{0,0}-u_{0,0} \mathcal E^{(1)}_{0,0} \frac{\partial \mathcal E^{(1)}_{0,0}}{\partial u_{0,0}}}{ \frac{\partial \mathcal E^{(1)}_{0,0}}{\partial u_{1,0}}}= \kappa_1 =- \frac{a^{(1)} \gamma^{(1)}_{1,1} + a^{(2)} \gamma^{(2)}_{1,1} - u_{1,0} \gamma^{(1)}_{1,0}}{\gamma^{(1)}_{1,0} }= \frac{d^{(1)} \gamma^{(1)}_{1,1} + d^{(2)} \gamma^{(2)}_{1,1} - u_{1,0} \gamma^{(2)}_{1,0} }{ \gamma^{(2)}_{1,0} },
 \eea
while from (\ref{g18}, \ref{g18a}) we get 
 \bea \label{g21}
 \frac{  \mathcal E^{(1)}_{0,0} \frac{\partial \mathcal E^{(1)}_{0,0}}{\partial u_{0,0}}}{ \frac{\partial \mathcal E^{(1)}_{0,0}}{\partial u_{1,0}}}= \kappa_2 =- \frac{c^{(1)} \gamma^{(1)}_{1,1} + c^{(2)} \gamma^{(2)}_{1,1} - u_{1,0} \gamma^{(1)}_{1,0}}{\gamma^{(1)}_{1,0} }= \frac{f^{(1)} \gamma^{(1)}_{1,1} + f^{(2)} \gamma^{(2)}_{1,1} - u_{1,0} \gamma^{(2)}_{1,0} }{ \gamma^{(2)}_{1,0} }.
 \eea
When $\frac{\partial \mathcal E^{(1)}_{0,0}}{\partial u_{1,0}}\ne 0$, solving the equations for $\mathcal E^{(1)}_{0,0}$ from (\ref{g19}, \ref{g20}, \ref{g21}) we get 
\bea \label{g22}
\mathcal E^{(1)}_{0,0} = \frac{\kappa_2 u_{0,0} + \kappa_1}{\kappa_0 - u_{10}}.
\eea
With no loss of generality we can set $\gamma^{(2)}=0$ and then $e^{(j)}=d^{(j)}=f^{(j)}=0$, $j=1,2$, $w^{(2)}=0$ and the equations (\ref{g19}, \ref{g20}, \ref{g21}) are compatible if $\kappa_2 \kappa_0 a^{(1)} = \kappa_1 b^{(1)} c^{(1)}$. The resulting class of linearizable P$\Delta$Es (\ref{g22}) is an extension of the Burgers equation (\ref{g8})
\bea \label{g23}
(\kappa_0-u_{1,0})(\kappa_2 u_{0,1} +\kappa_1) u_{0,0} - (\kappa_0-u_{1,1})(\kappa_2 u_{0,0} +\kappa_1) u_{1,0}=0,
\eea
which reduces to it when $\kappa_{0}\not=0$, $\kappa_1=1$ and $\kappa_2=0$. In (\ref{g23}) in all generality $\kappa_1$ can be taken to be either $0$ or $1$. Other two Burgers equations are obtained taking $\kappa_0\not=0$, $\kappa_1=0$ and $\kappa_2\not=0$ or $\kappa_0=0$, $\kappa_1=1$ and $\kappa_2\not=0$. All the these three Burgers equations can be transformed to $\left(1+u_{0,0}\right)u_{1,0}=\left(1+u_{0,1}\right)u_{0,0}$ and we recover the results obtained in \cite{ls122}. Moreover, if $\kappa_{2}\not=0$, $\kappa_1=1$  and $\kappa_{0}\not=0$, by the transformation $u_{0,0}=\kappa_{0} \frac{e_{1}-o_{2}}{e_{1}-e_{2}} \frac{\tilde u_{0,0}+e_{2}}{\tilde u_{0,0}+o_{2}}\;$, where $(e_j, o_j)$, $j=1,2$ are arbitrary parameters,  $\tilde u_{0,0}$ will satisfy the Hietarinta equation \cite{h} 
\begin{gather}\label{Rosa}
\frac{u_{0,0}+e_{2}} {u_{0,0}+e_{1}} \frac{u_{1,1}+o_{2}} {u_{1,1}+o_{1}}=\frac{u_{1,0}+e_{2}} {u_{1,0}+o_{1}} \frac{u_{0,1}+o_{2}} {u_{0,1}+e_{1}},
\end{gather}
with  
$\frac{(o_{1}-e_{2})(e_{1}-o_{2})}{(e_{1}-e_{2})(o_{1}-o_{2})}=-\kappa_{2} \kappa_{0}$.

When $ \frac{\partial \mathcal E^{(1)}_{0,0}}{\partial u_{1,0}}= 0$, we must have $\mathcal E^{(1)}_{0,0} \gamma^{(j)}_{1,0}=0$ which has no nontrivial solution. 
\section{Conclusions}   \label{five}
In this article we studied from the point of view of the symmetries the linearization through invertible or non--invertible transformations of nonlinear partial difference equations defined on a fixed non--transformable lattice. We find, in strict analogy to the continuous case, the linearizability conditions and apply them to the classification of linearizable P$\Delta$Es on a lattice of four points in the plane.  The results we obtain are compatible with the results obtained previously by requiring the existence of a linearizing transformation, however the classification turns out to be easier.

Work is in progress to extend these results to the case of linearizable equations belonging to more general plane lattices and to nonlinear partial difference equations defined on transformable lattices.

\section*{Acknowledgments}
We thank P. Winternitz for many enlightening discussions and the CRM, Universit\`e de Montr\`eal, for the hospitality of LD at the time when the present work was finished.

LD and SC have been partly supported by the Italian Ministry of Education and Research, 
 PRIN ``Continuous and discrete nonlinear integrable evolutions: from water
waves to symplectic maps" from 2010. 

\section*{Appendix}
In this Appendix we state and proof  a theorem on the necessary and sufficient conditions for a partial difference equation to be linear. For the sake of simplicity of the presentation  we limit our considerations to the case when the equation  is defined on 4 points (see Fig.1), i.e.
\bea \label{ap2}
\mathcal E \Big (n, m, u_{n,m}, u_{n+1,m},  u_{n,m+1}, u_{n+1,m+1} \Big )=0.
\eea
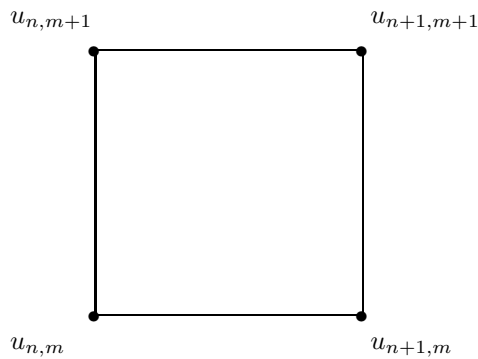
\begin{figure}[htbp] \label{fig1}
\begin{center}
\setlength{\unitlength}{0.1em}
\begin{picture}(200,140)(-50,-20)
 \put( 0,  0){\line(1,0){100}}
  \put( 100,  0){\line(-1,0){100}}
\put( 0,100){\line(1,0){100}}
  \put( 100,100){\line(-1,0){100}}
\put(  0, 0){\line(0,1){100}}
  \put(  0, 100){\line(0,-1){100}}
\put(100, 0){\line(0,1){100}}
  \put(100, 100){\line(0,-1){100}}
   \put(97, -3){$\bullet$}
   \put(-3, -3){$\bullet$}
   \put(-3, 97){$\bullet$}
   \put(97, 97){$\bullet$}
  \put(-32,-13){$u_{n,m}$}
  \put(103,-13){$u_{n+1,m}$}
  \put(103,110){$u_{n+1,m+1}$}
  \put(-32,110){$u_{n,m+1}$}
\end{picture}
\caption{Four points on a square}
\end{center}
\end{figure}

\begin{theorem} \label{at1}
Necessary and sufficient conditions for a discrete equation (\ref{ap2}) to have a symmetry of infinitesimal generator $\hat X= \phi_{n,m} \partial_{u_{n,m}}$, is that it is linear. The function $\phi_{n,m}$ must satisfy the following linear homogeneous discrete equation  
\bea \label{ap1}
\mathcal F \Big (\phi_{n,m}, \phi_{n+1,m},   \phi_{n,m+1}, \phi_{n+1,m+1} \Big )=\phi_{n+1,m+1} -a_{n,m} \phi_{n,m} - b_{n,m} \phi_{n,m+1} - c_{n,m} \phi_{n+1,m}=0.
\eea
\end{theorem}

\begin{proof}
It is almost immediate to prove that a linear partial difference equation defined on four lattice points \eqref{ap2} has a symmetry \eqref{e4} where $\phi_{n,m}$ satisfies a homogeneous linear equation. To obtain this result it is just sufficient to solve the invariance condition
\bea \label{ap2a}
\mbox{pr}\hat X \mathcal E \Bigl |_{\mathcal E=0}=0.
\eea

Not so easy is the proof that an equation which has such a symmetry \eqref{e4} must be linear. In the full generality (\ref{ap1}) can be rewritten  as
\bea \label{ap3}
\phi_{n+1,m+1}= F \Big (n, m, \phi_{n,m}, \phi_{n+1,m},  \phi_{n,m+1} \Big ),
\eea
and, by assumption, (\ref{ap2}) does not depend on $\phi_{n,m}$ and (\ref{ap3}) on $u_{n,m}$.

Let us prolong the symmetry generator (\ref{e4}) to all points contained in (\ref{ap2})
\bea \label{ap4}
\mbox{pr} \hat X = \phi_{n,m} \partial_{u_{n,m}} + \phi_{n+1,m} \partial_{u_{n+1,m}} + \phi_{n,m+1} \partial_{u_{n,m+1}} + \phi_{n+1,m+1} \partial_{u_{n+1,m+1}}.
\eea
The most generic equation (\ref{ap2}) having the symmetry (\ref{ap4}) will be written in terms of its invariants
\bea \label{ap5}
\mathcal E \Big ( n,m, K_1, K_2, K_3 \Big ) = 0,
\eea
with
\bea \label{ap6}
K_1 = \frac{u_{n,m+1}}{\phi_{n,m+1}}-\frac{u_{n,m}}{\phi_{n,m}}, \quad K_2 = \frac{u_{n+1,m}}{\phi_{n+1,m}}-\frac{u_{n,m}}{\phi_{n,m}}, \quad K_3 = \frac{u_{n+1,m+1}}{\phi_{n+1,m+1}}-\frac{u_{n,m}}{\phi_{n,m}}.
\eea
As (\ref{ap5}) depends   on $n,m$ we can with no loss of generality replace the invariants (\ref{ap6}) in (\ref{ap5}) by the functions 
\bea \label{ap7}
&&\tilde K_1 = u_{n,m+1}-u_{n,m}\frac{ \phi_{n,m+1}}{\phi_{n,m}}, \quad \tilde K_2 = u_{n+1,m}-u_{n,m} \frac{\phi_{n+1,m}}{\phi_{n,m}}, \\ \nonumber && \tilde K_3 = u_{n+1,m+1}-u_{n,m}\frac{ F \Big (n, m, \phi_{n,m}, \phi_{n+1,m},  \phi_{n,m+1} \Big )}{\phi_{n,m}}.
\eea 
Invariance of  (\ref{ap5}) then requires $\frac{\partial \mathcal E}{\partial \phi_{n,m}}=\frac{\partial \mathcal E}{\partial \phi_{n+1,m}}=\frac{\partial \mathcal E}{\partial \phi_{n,m+1}}=0$, i.e.
\bea \label{ap8a}
&&\frac{\partial \mathcal E}{\partial \tilde K_1} \Big ( u_{n,m} \frac{\phi_{n,m+1}}{\phi_{n,m}^2} \Big ) + \frac{\partial \mathcal E}{\partial \tilde K_2} \Big ( u_{n,m} \frac{\phi_{n+1,m}}{\phi_{n,m}^2} \Big )  + \frac{\partial \mathcal E}{\partial \tilde K_3} \Big ( \frac{F}{\phi_{n,m}^2} - \frac{ F_{,\phi_{n,m}}}{\phi_{n,m}} \Big )  u_{n,m}=0,
\\ \label{ap8b}
&&\frac{\partial \mathcal E}{\partial \tilde K_1} \Big ( -  \frac{u_{n,m}}{\phi_{n,m}} \Big ) + \frac{\partial \mathcal E}{\partial \tilde K_3} \Big (  - \frac{ F_{,\phi_{n,m+1}}}{\phi_{n,m}} \Big )  u_{n,m}=0,
\\ \label{ap8c}
&&\frac{\partial \mathcal E}{\partial \tilde K_2} \Big ( - \frac{u_{n,m}}{\phi_{n,m}} \Big ) + \frac{\partial \mathcal E}{\partial \tilde K_3} \Big ( - \frac{ F_{,\phi_{n+1,m}}}{\phi_{n,m}} \Big )  u_{n,m}=0.
\eea
As $\frac{\partial \mathcal E}{\partial \tilde K_j} \ne 0$ and (\ref{ap8a}, \ref{ap8b}, \ref{ap8c}) are a homogeneous system of algebraic equations, the determinant of the coefficients must be zero.  Consequently the function $ F$ must satisfy the first order linear partial differential equation
\bea \label{ap9}
 F -\phi_{n,m}  F_{,\phi_{n,m}} -\phi_{n,m+1}  F_{,\phi_{n,m+1}} -\phi_{n+1,m}  F_{,\phi_{n+1,m}}=0
\eea
i.e. $ F$ is given by
\bea \label{ap10}
 F = \phi_{n,m} f(\xi,\tau), \qquad \xi = \frac{\phi_{n+1,m}}{\phi_{n,m}}, \quad \tau = \frac{\phi_{n,m+1}}{\phi_{n,m}},
\eea where $f(\xi,\tau)$ is an arbitrary function of its arguments.

For $\ F$ given by (\ref{ap10}) the system (\ref{ap8a}, \ref{ap8b}, \ref{ap8c}) reduces to the following two equations for $\frac{\partial \mathcal E}{\partial \tilde K_j}$ 
\bea \label{ap11}
\frac{\partial \mathcal E}{\partial \tilde K_1}  + f_{\xi} \frac{\partial \mathcal E}{\partial \tilde K_3}=0,
\quad \frac{\partial \mathcal E}{\partial \tilde K_2}  + f_{\tau} \frac{\partial \mathcal E}{\partial \tilde K_3} =0,
\eea
whose solution is obtained by solving (\ref{ap11}) on the characteristics 
\bea \label{ap12}
\mathcal E &=& \mathcal E(n, m, L), \\ \nonumber L&=& u_{n+1,m+1} - f u_{n,m} - f_{,\xi} \Big (u_{n,m+1}-u_{n,m}\frac{ \phi_{n,m+1}}{\phi_{n,m}} \Big ) - f_{,\tau} \Big (  u_{n+1,m}-u_{n,m} \frac{\phi_{n+1,m}}{\phi_{n,m}} \Big ).
\eea
Requiring that $\mathcal E$ be independent of $\phi_{n,m}, \phi_{n+1,m}$ and $\phi_{n,m+1}$ we get $f_{,\xi \xi}=f_{,\tau \tau}=f_{,\xi \tau} =0$ i.e.
\bea \label{ap13a}
f &=& a_{n,m} + b_{n,m} \xi + c_{n,m} \tau, \quad  F = a_{n,m} \phi_{n,m} + b_{n,m} \phi_{n,m+1} + c_{n,m} \phi_{n+1,m}, \\ \label{ap13b}
L &=& u_{n+1,m+1} - a_{n,m} u_{n,m} - b_{n,m} u_{n,m+1} - c_{n,m} u_{n+1,m} .
\eea
$\mathcal E=0$ is an (non autonomous, maybe transcendental) equation for $L$ which, when solved, gives $L=d_{n,m}$, where $d_{n,m}$ stands for the set of the zeros of the equation (in addition to $n$ and $m$ possibly dependent on a set of parameters). In conclusion $u_{n,m}$ must satisfy  the linear equation   
\bea \label{ap14}
u_{n+1,m+1} - a_{n,m} u_{n,m} - b_{n,m} u_{n,m+1} - c_{n,m} u_{n+1,m} -d_{n,m}=0
\eea
\end{proof}
\begin{figure} \label{fig2}
\setlength{\unitlength}{1mm}
\begin{center}
\begin{picture}(100,51)
 \linethickness{1.pt}
     \put(80,10) {\circle*{1,5}}
  \put(50,10) {\circle*{1,5}}
   \put(50,40) {\circle*{1,5}}
 \put(20,10) {\circle*{1,5}}
\put(49,07){$u_{n,m}$}
\put(79,07){$u_{n+1,m}$}
\put(19,07){$u_{n-1,m}$}
\put(51,39){\boldmath $u_{n,m+1}$}
\put(40,50){\line(1,-1){55}}
\put(10,10) {\line(1,0){85}}
\put(50,10) {\line(0,1){40}}
\put(60,50){\line(-1,-1){55}}

\end{picture}
\end{center}
\caption[]{Four points on a triangle. }
\end{figure}
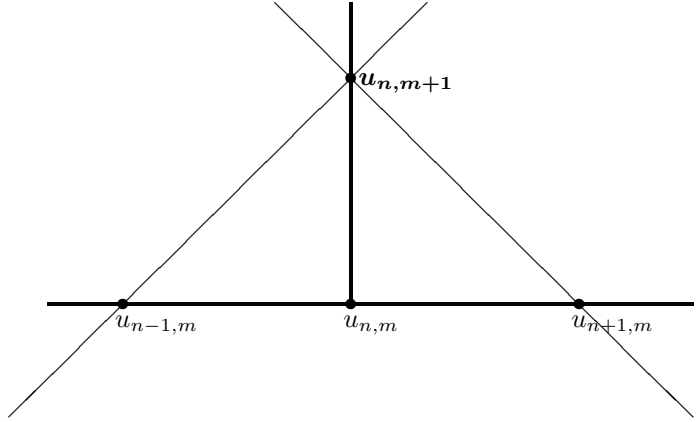
\begin{rem}
The proof of Theorem \ref{at1} does not depends on the position of the four lattice points considered in (\ref{ap2}). The same result is also valid if the four points are put on the triangle shown in Fig. 2, i.e.
\bea \label{ap2aa}
\mathcal E \Big (n, m, u_{n-1,m}, u_{n,m}, u_{n+1,m},  u_{n,m+1})=0.
\eea
\end{rem}



\end{document}